\let\orivec\vec
\documentclass{llncs} 

\usepackage{lmodern}
\usepackage[T1]{fontenc}
\usepackage[utf8]{inputenc}

\normalfont %to load T1lmr.fd
\DeclareFontShape{T1}{lmr}{bx}{sc}{ <-> ssub * cmr/bx/sc }{}

\let\springervec\vec
\let\vec\orivec
\usepackage{amsmath}
\let\vec\springervec

\usepackage{amsfonts}
\usepackage{amssymb}
\usepackage{graphicx}
\usepackage[ruled,vlined]{algorithm2e}
\usepackage{array}
\usepackage{tabu}
\usepackage{cite}
\usepackage{enumitem}
\usepackage{placeins}

\newcommand{\ind}{\ensuremath{\textsc{Ind}}}
\newcommand{\con}{\ensuremath{\textsc{Con}}}
\newcommand{\clique}{\ensuremath{\textsc{Clique}}}
\newcommand{\stcut}{\ensuremath{s\mbox{-}t\mbox{-}\textsc{Cut}}}

\newcommand{\msum}{\ensuremath{\textsc{Sum}}}
\newcommand{\mmax}{\ensuremath{\textsc{Max}}}
\newcommand{\mnum}{\ensuremath{\textsc{Num}}}

\newcommand{\Opt}{\ensuremath{\texttt{Opt}}}
\newcommand{\Min}{\ensuremath{\texttt{Min}}}

\renewcommand{\P}{\ensuremath{\textsf{\upshape P}}}
\newcommand{\NP}{\ensuremath{\textsf{\upshape NP}}}

\newcommand{\True}{\ensuremath{\texttt{True}}}
\newcommand{\False}{\ensuremath{\texttt{False}}}

\begin{document}

%For arxiv
\renewcommand{\O}{\ensuremath{O}}

\title{Exact and approximate algorithms for movement problems on (special classes of) graphs
\thanks{Research partially supported by the Research Grant PRIN 2010 ``ARS
TechnoMedia'' funded by the Italian Ministry of University and Research.}
\thanks{A preliminary version of this work appeared in the \emph{Proceedings of the 20th Colloquium on Structural Information and Communication Complexity (SIROCCO'13)}, LNCS 8179, Springer, 322--333, 2013. DOI: http://dx.doi.org/10.1007/978-3-319-03578-9\_27}
}

\author{Davide Bil\`o\inst{1} \and
Luciano Gual\`a\inst{2} \and Stefano Leucci\inst{3} \and Guido Proietti\inst{3,4}
\institute{Dipartimento di Scienze Umanistiche e Sociali,
Università di Sassari, Italy \and Dipartimento di Ingegneria dell'Impresa,
Università di Roma ``Tor Vergata", Italy \and Dipartimento di
Ingegneria e Scienze dell'Informazione e Matematica, \\Università degli Studi dell'Aquila, Italy  \and Istituto di Analisi dei Sistemi
  ed Informatica,
  CNR, Roma, Italy \\
E-mail: \texttt{davide.bilo@uniss.it; guala@mat.uniroma2.it; stefano.leucci@univaq.it; guido.proietti@univaq.it}
}}
           
\date{Received: date / Accepted: date}

\maketitle

\begin{abstract}
When a large collection of objects (e.g., robots, sensors, etc.) has to be deployed in a given environment, it is often required to plan a coordinated motion of the objects from their initial position to a final configuration enjoying some global property. In such a scenario, the problem of minimizing some function of the distance travelled, and therefore energy consumption, is of vital importance. In this paper we study several motion planning problems that arise when the objects must be moved on a graph, in order to reach certain goals which are of interest for several network applications. Among the others, these goals include broadcasting messages and forming connected or interference-free networks.
We study these problems with the aim of minimizing a number of natural measures such as the average/overall distance travelled, the maximum distance travelled, or the number of objects that need to be moved. To this respect, we provide several approximability and inapproximability results, most of which are tight.
\end{abstract}

\section{Introduction}

In many practical applications a number of centrally controlled objects need to be moved in a given environment in order to complete some task. Problems of this kind often occur in robot motion planning where we seek to move a set of robots from their starting position to a set of ending positions such that a certain property is satisfied.
For example, if the robots are equipped with a short range communication device we might want to move them so that a message originating from one of the robots can be routed to all the others. If the robots' goal is to monitor a certain area we might want to move them so that they are not too close to each other. Other interesting problems include gathering (placing robots next to each other), monitoring of traffic between two locations, building interference-free networks, and so on.
To make things harder, objects to be moved are often equipped with a limited supply of energy. Preserving energy is a critical problem in
ad-hoc networking, and movements are expensive. To prolong the lifetime of the objects we seek to minimize the energy consumed during movements and thus the distance travelled.
Sometimes, instead, movements are cheap but before and/or after an object moves it needs to perform expensive operations. In this scenario we might be interested in moving the minimum number of objects needed to reach the goal.

In this paper, we assume the underlying environment is actually a \emph{network}, which can be modelled as an undirected graph $G$, and the moving objects are centrally controlled \emph{pebbles} that are initially placed on vertices of $G$, and that can be moved to other vertices by traversing the graph edges. To this respect, we study several movement planning problems that arise by various combinations of final positioning goals and movement optimization measures.
In particular, we focus our study on the scenarios where we want the pebbles to be moved to a \emph{connected subgraph} (\con), an \emph{independent set} (\ind), or a \emph{clique} (\clique) of $G$, while minimizing either the \emph{overall movement} (\msum), the \emph{maximum movement} (\mmax), or the \emph{number of moved pebbles} (\mnum). We also give some preliminary results on the problem of moving the pebbles to an \emph{s-t-cut}, i.e., a set of vertices whose removal makes two given vertices $s,t$ disconnected (\stcut) while minimizing the above measures.

We will denote each of the above problems with $\psi$-$c$, where $\psi$ represents the goal to be achieved and $c$ the measure to be minimized. For a more rigorous definition of the problems we refer the reader to Section \ref{sec:formal_definition}.

\paragraph{Related work.}
Although movement problems were deeply investigated in a distributed setting (see \cite{PS06} for a survey), quite surprisingly the centralized counterpart has received attention from the scientific community only in the last few years.

The first paper which defines and studies these problems in this latter setting is \cite{demaine2007minimizing}. In their work, the authors study the problem of moving the pebbles on a graph $G$ of $n$ vertices so that their final positions form a \emph{connected component}, a \emph{path} (directed or undirected) \emph{between two specified nodes}, an \emph{independent set}, or a \emph{matching} (two pebbles are matched together if their distance is exactly $1$). 

Regarding connectivity problems, in \cite{demaine2007minimizing} the authors show that all the variants are hard and that the approximation ratio of $\con$-$\mmax$ is between $2$ and $\O(1+\sqrt{k/c^*})$, where $k$ is the number of pebbles and $c^*$ denotes the measure of an optimal solution. This result has been improved in \cite{berman2011}, where the authors show that $\con$-$\mmax$ can be approximated within a constant factor. In \cite{demaine2007minimizing} it is also shown that $\con$-$\msum$ and $\con$-$\mnum$ are not approximable within $\O(n^{1-\epsilon})$ (for any positive $\epsilon$) and $o(\log n)$, respectively, while they admit approximation algorithms with ratios of $\O(\min\{n \log n, k\})$ and $\O(k^\epsilon)$, respectively. Moreover, the authors also provide an exact polynomial-time algorithm for $\con$-$\mmax$ on trees.

Concerning independency problems, in \cite{demaine2007minimizing} the authors remark that it is \NP-hard even to find any feasible solution on general graphs since it would require to find an independent set of size at least $k$. This clearly holds for all three objective functions. For this reason, they study an Euclidean variant of these problems where pebbles have to be moved on a plane so that their pairwise distances are strictly greater than $1$. In this case, the authors provide an approximation algorithm that guarantees an additive error of at most $1+1/\sqrt{3}$ for $\ind$-$\mmax$, and a polynomial time approximation scheme for $\ind$-$\mnum$.

More recently, in \cite{friggstad2011minimizing}, a variant of the classical facility location problem has been studied. This variant, called \emph{mobile facility location}, can be modelled as a movement problem and is approximable within $(3+\epsilon)$ (for any constant $\epsilon>0$) if we seek to minimize the total movement \cite{ahmadian2013local}, while the variant where the maximum movement has to be minimized admits a tight $2$-approximation \cite{demaine2007minimizing, friggstad2011minimizing}. Moreover, as it is frequent in the practice to have a small number of pebbles compared to the size of the environment (i.e., the vertices of the graph), the authors of~\cite{demaine2009FPT} turn to study fixed-parameter tractability. They show a relation between the complexity of the problems and their \emph{minimal configurations} (sets of final positions of the pebbles that correspond to feasible solutions, such that any removal of an edge makes them unacceptable). Finally, we mention that in \cite{BDGMPW13} it was considered a set of vertex-to-vertex motion planning problems in a simple polygon, with the aim of forming final configurations enjoying some sort of \emph{visual connectivity} among the pebbles.

\paragraph{Our results.}

We start by studying connectivity motions problems in the case where pebbles move on a tree, and we devise two polynomial-time dynamic programming algorithms for $\con$-$\msum$ and $\con$-$\mnum$. These algorithms complement the already known polynomial-time algorithm for $\con$-$\mmax$ on trees shown in \cite{demaine2007minimizing}.

Then, we study independency motion problems on graphs where a \emph{maximum independent set} (and thus a feasible solution for the corresponding motion problem) can be computed in polynomial time. This class of graphs includes, for example, perfect and claw-free graphs.
More precisely, we show that $\ind$-$\mmax$ and $\ind$-$\msum$ are \NP-hard even on bipartite graphs (which are known to be perfect graphs \cite{bollobas1998modern}).
Moreover, we devise three exact polynomial-time algorithms: one for solving $\ind$-$\mmax$ on paths, and the other two for solving $\ind$-$\msum$ and $\ind$-$\mnum$ on trees, respectively.
Moreover, we devise a polynomial-time approximation algorithm for $\ind$-$\mmax$ which is optimal unless an additive term of $1$ (this is clearly tight).

Concerning the problem of moving pebbles towards a clique of a general graph, we prove that all the three variants are \NP-hard. Then, we provide an approximation algorithm for $\clique$-$\mmax$ which is optimal unless an additive term of $1$ (this result is clearly tight). Moreover, we show that both $\clique$-$\msum$ and $\clique$-$\mnum$ are approximable within a factor of $2$, but they are not approximable within a factor better than $10\sqrt{5}-21 > 1.3606$, unless $\P=\NP$.
If the \emph{unique game conjecture} \cite{khot2002power} is true, then both problems are not approximable within a factor better than $2$ and the provided approximation algorithms are tight.
These results are obtained by showing a non-trivial relation with the \emph{minimum vertex cover} problem.
We also show that an exact solution for $\clique$-$\mnum$ can be computed in polynomial time on every class of graphs for which finding a \emph{maximum-weight clique} requires polynomial time (these classes of graphs also include perfect and claw-free graphs).

Finally, we present a strong inapproximability results of $\Omega(n^{1-\epsilon})$ (for any $\epsilon>0$) for $\stcut$-$\mmax$ and $\stcut$-$\msum$, unless $\P=\NP$, along with two approximation algorithms. The approximation algorithm for  $\stcut$-$\mmax$ is essentially tight, while we show that any constant-factor approximation for $\stcut$-$\mnum$ would imply a tight approximation for $\stcut$-$\msum$.

The paper is organized as follows: in Section~\ref{sec:formal_definition} we provide a formal definition of our problems, while in Sections~\ref{sec:con}--\ref{sec:stcut} we give our results for \con, \ind, \clique, and \stcut, respectively (for a summary of the state of the art of the studied problems, along with the results presented in this paper, see Table \ref{table:results}). Finally, Section~\ref{sec:concl} concludes the paper.

\begin{table}[!ht]
\centering
\scriptsize
\setlength{\tabulinesep}{.5mm} %vertical padding
\setlength{\tabcolsep}{.5mm} %horizonal padding
\newcolumntype{L}{>{\raggedright\arraybackslash}p{3.58cm}}
\newcolumntype{M}{>{\raggedright\arraybackslash}p{3.40cm}}
\begin{tabu}[t]{|c|M|L|M|}
\hline\everyrow{\hline}
 & $\mmax$ & $\msum$ & $\mnum$ \\
$\con$ &
G: $2 \le \rho = \O(1)$ \hfill\cite{demaine2007minimizing, berman2011} \linebreak T: polynomial \hfill\cite{demaine2007minimizing} &
G: $\rho = \Omega(n^{1-\epsilon})$ \hfill\cite{demaine2007minimizing} \linebreak \phantom{G:} $\rho=\O(\min\{n \log n,k\})\!$ \hfill\cite{demaine2007minimizing} \linebreak \textbf{T: polynomial} &
G: $\rho=\Omega(\log n)$ \hfill\cite{demaine2007minimizing} \linebreak \phantom{G:} $\rho=\O(k^\epsilon)$ \hfill\cite{demaine2007minimizing} \linebreak \textbf{T: polynomial} \\
$\ind$ &
G: \NP-hard \hfill\cite{demaine2007minimizing} \linebreak \textbf{IS:} {\boldmath$c^* + 1$}\textbf{,} {\boldmath$\rho \le 2$} \linebreak \textbf{B:} {\boldmath$\rho \ge 2$}  \linebreak \textbf{P: polynomial} &
G: \NP-hard \hfill\cite{demaine2007minimizing} \linebreak \textbf{B: \NP-hard} \linebreak \textbf{T: polynomial} &
G: \NP-hard \hfill\cite{demaine2007minimizing} \linebreak \textbf{T: polynomial} \\
$\clique$ &
\textbf{G: \NP-hard} \linebreak \phantom{\textbf{G:}} {\boldmath$c^*+1$} &
\textbf{G:} {\boldmath$10\sqrt{5}-21 \le \rho \le 2$}  &
\textbf{G:} {\boldmath$10\sqrt{5}-21 \le \rho \le 2$} \linebreak \textbf{MWC: polynomial} \\
$\stcut$ &
\textbf{G:} {\boldmath$\rho = \Omega(n^{1-\epsilon})$} \linebreak \phantom{\textbf{G:}} {\boldmath$\rho \le d$} &
\textbf{G:} {\boldmath$\rho = \Omega(n^{1-\epsilon})$} \linebreak \phantom{\textbf{G:}} {\boldmath$\rho \le k \cdot d$} &
\textbf{G:} \textbf{\boldmath$\rho$-apx $\Longrightarrow (\rho \cdot d)$-apx for $\stcut$-$\msum$} \\
\end{tabu}
\ \\
\ \\

\label{table:results}
\caption{Known and new (in bold) results for the various motion problems on general graphs (G), bipartite graphs (B), graphs on which a maximum independent set or a maximum-weight clique can be computed in polynomial time (IS, MWC), trees (T), and paths (P). With $n$ and $d$ we denote the number of vertices and the diameter of $G$, respectively, while $k$ denotes the number of pebbles, $\rho$ denotes the best approximation ratio for the corresponding problem, and finally $c^*$ is the measure of an optimal solution. Notice that for independency problems on general graphs it is \NP-hard even to find any feasible solution. All the inapproximability results hold under the assumption that $\P\not=\NP$.}
\end{table}

\section{Formal definitions}
\label{sec:formal_definition}

	A pebble motion problem, denoted by $\psi$-$c$, is an optimization problem whose instances consist of a loop-free connected undirected graph $G=(V(G),E(G))$ on $n$ nodes, a set $P=[k]=\{1,\dots, k\}$ of \emph{pebbles}, a function $\sigma : P \to V(G)$ that assigns each pebble to a \emph{start vertex} of $G$, and a boolean predicate $\psi : 2^{V(G)} \to \{ \True,\False \}$ that assigns a truth value to every possible subset of vertices of $G$.
	
	A (feasible) solution is a function $\mu : P \to V(G)$ that maps each pebble to an \emph{end vertex} of $G$ (in other words, \emph{moves} a pebble from its start to its end position) such that $\psi(\mu[P])$ is true, where $\mu[P]$ denotes the image of $P$ under $\mu$.
	Notice that, in general, it is not required for $\sigma$ or $\mu$ to be injective and thus we allow more than one pebble to be placed on the same vertex. In the rest of the paper, we will assume that a pebble moving from a vertex $u$ to a vertex $v$ always uses a \emph{shortest path} in $G$ between $u$ and $v$, say $\pi_G(u,v)$. Moreover we denote by $d_G(u,v)$ the length of such a path.
	Finally, $c(\mu) \in \mathbb{N}_0$ is a measure function that assigns a non-negative integer to each feasible solution (i.e., to each set of moves). A solution $\mu^*$ that minimizes $c$ is said to be \emph{optimal}.
	
	In the following, we will study some of the movement problems that arise from the different choices of predicates and measures. 	
		In particular, we will consider the following predicates:
	\begin{description}
		\item[\emph{Connectivity:}] $\con(U)$ is true if and only if the subgraph of $G$ induced by the set of vertices $U \subseteq V(G)$ is connected;
		\item[\emph{Independency:}] $\ind(U)$ is true if and only if $U \subseteq V(G)$ is an independent set of $G$ of size $k$, i.e., there is at most one pebble per vertex and no two pebbles are on adjacent vertices;
		\item[\emph{Clique:}] $\clique(U)$ is true if and only if $U  \subseteq V(G)$ induces a clique in $G$, i.e., for each pair $u,v$ of distinct vertices in $U$ there exists the edge $(u,v) \in E(G)$;
		\item[\emph{$s$-$t$-Cut:}] Given $s,t \in V(G)$ with $s \not= t$, then $\stcut(U)$ is true if and only if $s \not \in U$, $t \not\in U$ and $U  \subseteq V(G)$ is an \emph{s-t}-cut (i.e., there exists no path between $s$ and $t$ in the graph induced by the vertices in $V(G) \setminus U$);
	\end{description}		
	\noindent and the following measures:
	\begin{description}
		\item[\emph{Overall movement:}] The sum of the distances travelled by pebbles has to be minimized: every pebble $p \in P$ moves from its starting vertex $\sigma(p)$ to his end vertex $\mu(p)$, so the overall distance is $\msum(\mu) = \sum_{p\in P} d_G( \sigma(p), \mu(p) )$;
		\item[\emph{Maximum movement:}] We want to minimize the maximum distance travelled by a pebble, i.e., the measure $\mmax(\mu)= \max_{p \in P} d_G( \sigma(p), \mu(p) )$;
		\item[\emph{Number of moved pebbles:}] We aim to minimize the number of pebbles that need to be moved from their starting positions. The associated measure is $\mnum(\mu) = | \{ p \in P : \sigma(p) \not= \mu(p)\} |$.
	\end{description}	

	\section{Connectivity motion problems}
	\label{sec:con}

	In this section we describe two polynomial-time algorithms for solving on trees $\con$-$\msum$ and $\con$-$\mnum$, respectively. In this way we complement the result provided in \cite{demaine2007minimizing} for $\con$-$\mmax$ on trees.

	\subsection{Solving $\con$-$\msum$ on trees}	
	\label{sec:con_sum_trees}

	Our dynamic-programming algorithm relies on the following property of optimal solutions:
	\begin{lemma}
		\label{lemma:sum_tree_no_crossing}
		In any optimal solution $\mu^*$ for an instance of $\con$-$\msum$ on trees, there exists no edge that is traversed in opposite directions by pebbles.
	\end{lemma}
	\begin{proof}
		Let $\pi_G(u,v)$ denote a shortest path in $G$ between the vertices $u$ and $v$.
		Suppose by contradiction that there exists an optimal solution $\mu^*$, an edge $(x,y) \in E(G)$, and two pebbles $p, q \in P$ such that $p$ moves through the path $\pi_G(\sigma(p), x) \cup \{ (x,y) \} \cup \pi_G(y, \mu^*(p))$ and $q$ moves through the path $\pi_G(\sigma(q), y) \cup \{ (y,x) \} \cup \pi_G(x, \mu^*(q))$.
		Consider the solution $\mu^\prime$ obtained from $\mu^*$ by swapping the final positions for $p$ and $q$, i.e., $\mu^\prime(p)=\mu^*(q)$ and $\mu^\prime(q)=\mu^*(p)$.
Clearly $\mu^\prime[P]=\mu^*[P]$, therefore $\mu^\prime$ is feasible. Moreover $c(\mu^\prime) < c(\mu^*)$ as we can move $p$ through the path $\pi_G(\sigma(p), x) \cup \pi_G(x,\mu^\prime(p))$ and $q$ through the path $\pi_G(\sigma(q), y) \cup \pi_G(y, \mu^\prime(q))$, thus saving $2$.
	\qed\end{proof}

	The algorithm first guesses a vertex $r \in V(G)$ such that there exists an optimal solution that places a pebble on $r$, then roots the tree $G$ at $r$ to obtain a rooted tree $G_r$, and finally considers all the subtrees of $G_r$ in a bottom-up fashion.

	For a given subtree $T_u$ of $G_r$ rooted at the vertex $u$, let us denote by $\eta(u)$ the number of pebbles placed on $V(T_u)$ w.r.t. $\sigma$.
	When the subtree $T_u$ is examined, we consider an auxiliary problem. In this problem we want to place exactly $j \le k$ of the pebbles on the vertices of $T_u$ in order to satisfy the following properties:
	\begin{enumerate}[label=(P\arabic*),leftmargin=*]
		\item \label{it:con_P1} the subgraph of $T_u$ induced by the final positions of the pebbles must be connected;
		\item \label{it:con_P2} if $j \ge 1$ then at least one pebble must be placed on $u$.
	\end{enumerate}

	Moreover, if $j < \eta(u)$ we want to move the $\eta(u)-j$ exceeding pebbles to the parent of $u$ (and thus outside $T_u$). In a similar manner, if $j > \eta(u)$ then the $j-\eta(u)$ missing pebbles are to be moved into $T_u$ from the parent of $u$, where we assume they are initially placed. We point out that, by Lemma \ref{lemma:sum_tree_no_crossing}, we do not need to consider the case where some pebbles move into $T_u$ while others move out of $T_u$.

	\begin{figure}[t]
		\centering
		\includegraphics[scale=1.3]{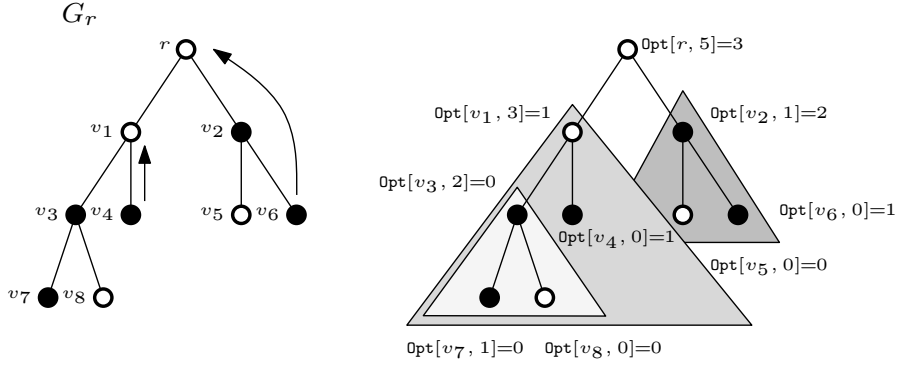}
		\caption{Left: an instance of $\con$-$\msum$ and its optimal solution. A single pebble is placed on each black vertex while the optimal movements are denoted by arrows. Right: the auxiliary problems corresponding to the optimal solution which are considered by the dynamic-programming algorithm.}
		\label{fig:mov_con_sum}
	\end{figure}

	We will denote by $\Opt[u,j]$ the cost of the optimal movement for this auxiliary problem. Notice that, in $\Opt[u,j]$, we are accounting for the cost of traversing all the edges of $T_u$ plus the edge from $u$ to its parent. To solve the original problem we need to find a solution corresponding to $\Opt[r,k]$. Clearly as $k=\eta(r)$ we do not have exceeding or missing pebbles, in this case. We now show how to combine these auxiliary problems.

	If $u$ is a leaf of $G_r$ then $\Opt[u,j]=|\eta(u) - j|$. Otherwise, if $u$ is not a leaf, we can distinguish two cases: $j=0$ and $j > 0$.
	If $j=0$ then no pebble can be placed on $u$ or in any descendant of $u$, therefore all the pebbles must first be moved towards $u$ and then to the parent of $u$. Let $v_1, \dots, v_\ell$ be the set of children of $u$ in $G_r$, we have: $\Opt[u,0] = \eta(u) + \sum_{i=1}^\ell \Opt[v_i, 0]$.
	
	Otherwise, if $j > 0$, we can move any number of pebbles between $1$ and $j$ to $u$, and place the remaining pebbles on the subtrees rooted at the children of $u$. Therefore we have:
	\[
		\Opt[u,j] = |\eta(u)-j| + \min_{ \substack{0 \le j_1, \dots, j_\ell < j \\ \sum_{i=1}^{\ell} j_i < j } } \left\{ \sum_{i=1}^{\ell} \Opt[v_i, j_i] \right\} \mbox{.}
	\]
	
	Notice that the minimum considers all the possible ways for distributing less than $j$ pebbles on the subtrees, i.e., all the vectors $(j_i)_i$ of $\ell$ elements whose sum is less than $j$.
	
	We now argue on the fact that, despite the number of such vectors can be exponential on $j$, the minimum can be found in polynomial time.
	
	This can again be done by using dynamic programming: let $\Min[i,h]$ denote the minimum cost of placing $h$ pebbles in the first $i$ subtrees. Clearly when $i=1$ we have $\Min[1, h] = \Opt[v_1, h]$, while for $i > 1$ the following holds:
	\[	
	 	\Min[i, h] = \min_{0 \le z \le h} \left\{ \Min[i-1, z] + \Opt[v_i, h-z] \right\} \mbox{.}
	\]

	\noindent 
Therefore, the equation for $\Opt[u,j]$ can be rewritten as:	
	\[
		\Opt[u,j] = |\eta(u)-j| + \min_{0 \le i < j} \Min[\ell, i] \mbox{.}
	\]	
	
	Notice how this way of distributing the pebbles is general and does not depend on the specific movement problem: in fact, it can be used every time we are interested in minimizing the cost of distributing a number of items in a set of bins if, for each bin, we incur a cost that depends on the number of items placed therein.
	
	An example of an optimal decomposition into subproblems along with the corresponding optimal solution is shown in Figure~\ref{fig:mov_con_sum}.
	
	Regarding the complexity of the algorithm, the time required to compute a specific $\Opt[u,j]$ is $O(\ell \cdot j) = O(\ell \cdot k)$. As the sum of the $\ell$-values over all the vertices is $n-1$, the time needed to compute $\Opt[v,j]$ for a fixed $j$ and all $v \in V$ is $O(n \cdot k)$. It follows that all the possible subproblems can be solved in time $O(n \cdot k^2)$.

	As we have to guess the vertex $r$, a na\"ive strategy would be repeating the above procedure $n$ times, one for each vertex of $G$. This would require an overall time of $O(n^2 \cdot k^2)$.
	We can do better by using a more sophisticated approach: consider a \emph{centroid}\footnote{A centroid of a tree is a vertex whose removal minimizes the maximum number of nodes over all the trees of the resulting forest. Notice that each tree of the forest has at most half of the vertices of the original tree, and that a centroid can be easily found in linear time.} $v$ of $G$, and notice that either there exists an optimal solution that places a pebble on $v$, or every optimal solution places all the pebbles on a single connected component of $G-v$.
	We first apply the above algorithm using $v$ as the root and then we proceed recursively on the trees of the forest $G-v$ (each of which has at most half of the vertices).
	More precisely, for every subtree $T$ of $G-v$, rooted at $v^\prime$, we recursively solve an instance consisting of the tree $T$ where all the pebbles in $V(T)$ are left unmoved and all pebbles not in $V(T)$ have been moved to $v^\prime$. This movement cost, i.e., $\sum_{p \in P : \sigma(p) \not\in V(T)} d(\sigma(p), v^\prime)$, is then added to the measure of the solution returned by the recursive call.
	Among all the computed solutions we choose the cheapest one.
	
	By doing so, we are able to reduce the computational complexity to \linebreak $O(n \cdot k^2 \log n)$. Indeed, the recurrence relation describing the running time of the algorithm is $T(n) = \sum_{n_j} T(n_j) + O(n \cdot k^2)$, where $n_j \le \frac{n}{2}$ denotes the number of vertices of the $j$-th subtree of $G-v$. Clearly, the depth of the recursion is $O(\log n)$ while the amount of work on each level of the recursion-tree is $O(n \cdot k^2)$.
	
	Once the value of the optimal solution $\mu^*$ has been found, it is not too hard to see that the optimal solution itself can be reconstructed by proceeding in a bottom-up fashion, while keeping track of both the pebbles that move out of each subtree and the position where missing pebbles are to be placed. 

	To summarize, we have the following:
	\begin{theorem}
		\label{thm:con_sum_trees}
		$\con$-$\msum$ on trees can be solved in $O(n \cdot k^2 \log n)$ time.
	\end{theorem} %

	\subsection{Solving $\con$-$\mnum$ on trees}

	The algorithm is similar to the one for $\con$-$\msum$: we guess a vertex $r \in V(G)$ such that there exists an optimal solution that places a pebble on $r$, then we root the tree $G$ at $r$ (call $G_r$ the rooted tree) and we consider all the subtrees of $G_r$ in a bottom-up fashion.

	Let $\varphi(u)=|\{ p \in P : \sigma(p)=u \}|$ be the number of pebbles whose initial position is the vertex $u \in V(G)$.

	As before, when the subtree $T_u$ (rooted at the vertex $u$) of $G_r$ is examined, we consider an auxiliary problem where we want to place exactly $j \le k$ pebbles on the vertices of $T_u$ in order to satisfy the properties \ref{it:con_P1} and \ref{it:con_P2}.

	We will measure the cost of a solution for this auxiliary problem by examining the number of pebbles placed on each vertex of $T_u$. Removing pebbles from a vertex costs nothing, while placing a pebble on a vertex costs $1$ if it comes from a different vertex.	A way to visualize this auxiliary problem is to imagine the tree $T_u$ where no pebbles have been placed and	a pool of $j$ pebbles to be distributed on its vertices. Each vertex $v$ of $T_u$ can hold up to $\varphi(v)$ pebbles for free, while each additional pebble placed on $v$ increases the overall cost by $1$.

	We will denote by $\Opt[u,j]$ the cost of the optimal movement for this auxiliary problem. To solve the original problem we need to find a solution corresponding to $\Opt[r,k]$. We now show how to combine these auxiliary problems.

	If $u$ is a leaf of $G_r$ then we have:
	\[
		\Opt[u,j] =
		\begin{cases}
			0 		& \mbox{if } j \le \varphi(u); \\
			j-\varphi(u) 	& \mbox{if } j > \varphi(u).
		\end{cases}	
	\]

	Otherwise, if $u$ is not a leaf, we can either place some pebbles on $u$ and the others on the subtrees rooted at its children ($j \ge 1$), or place no pebble at all in the whole subtree rooted at $u$ ($j=0$). We call $z$ the number of pebbles that are to be placed on $u$.
	
	If $j=0$ we have $\Opt[u,j] = 0$, otherwise:
	\[
		\Opt[u,j] = \min_{1 \le z \le j}\left\{ \max \{z-\varphi(u), 0\} + \min_{\substack{0 \le j_1, \dots, j_{\ell} \le j-z \\ \sum_{i=1}^{\ell} j_i = j-z }} \left\{ \sum_{i=1}^{\ell} \Opt[v_i, j_i] \right\} \right\}
	\]
	where $v_1, \dots, v_\ell$ are the children of $u$ in $G_r$.
	
	As before, using the already shown dynamic-programming approach to optimally distributing the pebbles on the subtrees, we can find the values of $\Opt[u,j]$ for a fixed $j$ and all $v \in V$, in $O(n \cdot k)$ time. Therefore the time required to compute all $\Opt$ values for a single root $r$ is $O(n \cdot k^2)$ and the measure of the best solution is found in $\Opt[r, k]$.
	As for $\con$-$\msum$, it is not necessary to run the algorithm for all roots $r \in V(G)$ but we can choose a centroid $v$ of $G$ as starting root and then proceed recursively on the trees of $G-v$.

	To summarize, we have the following:
	\begin{theorem}
		\label{thm:con_num_trees}
		$\con$-$\mnum$ on trees can be solved in $O(n \cdot k^2 \log n)$ time.
	\end{theorem} %

	\section{Independency motion problems}
	\label{sec:ind}

	In this section we focus on independency motion problems. First, we give a better characterization of the hardness of $\ind$-$\mmax$ and $\ind$-$\msum$ (depending on the input graph), and then we show some positive results for our considered variants on paths and trees. Since if $k \geq n$ there is no feasible solution, we will consider only instances where $k < n$.

	\subsection{Hardness of $\ind$-$\mmax$ and $\ind$-$\msum$ on bipartite graphs}
	As we already pointed out, for independency problems on general graphs it is \NP-hard even to find any feasible solution since it would require to find an independent set of size at least $k$. Nevertheless, one may wonder whether independency motion problems are tractable on instances on which a maximum independent set can be found in polynomial time. We provide a negative answer to this question, at least for $\ind$-$\mmax$ and $\ind$-$\msum$, by showing the following 

	\FloatBarrier
	\begin{figure}[t]
		\centering
		\includegraphics[scale=1.3]{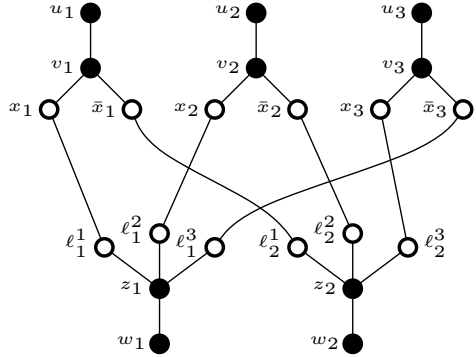}
		\caption{Instance of $\ind$-$\mmax$ and of $\ind$-$\msum$ corresponding to the formula $(\bar{x}_1 \vee \bar{x}_2 \vee x_3) \wedge (x_1 \vee x_2 \vee \bar{x}_3)$. Pebbles are placed on black vertices.}
		\label{fig:mov_ind_max}
	\end{figure}

	\begin{theorem}
		\label{thm:ind_max_sum_hard_bipartite}
		$\ind$-$\mmax$ and $\ind$-$\msum$ are \NP-hard on bipartite graphs.
	\end{theorem}
	\begin{proof}
		We will show a polynomial reduction from $\textsc{3-Sat}$ to the decisional versions of $\ind$-$\mmax$ and $\ind$-$\msum$.
			Recall that $\textsc{3-Sat}$ is the problem of deciding whether a formula $f$ in conjunctive normal form with three literals per clause is satisfiable. Let $X = \{ x_1, \dots, x_\tau \}$ be the set containing the variables of $f$, and let $m$ be the number of clauses. We will denote the $i$-th literal of the $j$-th clause with $\ell^i_j$.
			
		Given an instance $f$ for $\textsc{3-Sat}$ we construct an instance for $\ind$-$\mmax$ and for $\ind$-$\msum$ in the following manner:
	\begin{itemize}
		\item For each variable $x_i \in X$ create a star with $3$ leaves labelled $u_i, x_i, \bar{x}_i$ and label the internal node $v_i$. Place one pebble on $u_i$ and one on $v_i$.
			
		\item For each clause $(\ell^1_j \vee \ell^2_j \vee \ell^3_j)$ create a star with $4$ leaves labelled $\ell^1_j, \ell^2_j, \ell^3_j, w_j$ and label the internal node $z_j$. Place one pebble on $w_j$ and one on $z_j$.
			
		\item For each literal $\ell^i_j$ of $f$ let $x_s$ be the corresponding variable; then, if $\ell^i_j$ is asserted add an edge between the two nodes labelled $\ell^i_j$ and $\bar{x}_s$, otherwise add an edge between the two nodes labelled $\ell^i_j$ and $x_s$.
	\end{itemize}					
		
	Let $G$ be the resulting graph, $P$ the set containing the $2 \cdot (\tau+m)$ placed pebbles, and $\sigma$ the function that maps each pebble to its starting position (see Figure \ref{fig:mov_ind_max}).
	Notice that $G$ is bipartite as we can partition the vertices into two sets
	$A=\{v_i : 1 \le i \le \tau\} \cup \{ \ell^i_j : 1 \le i \le 3 \wedge 1 \le j \le m \} \cup \{ w_j : 1 \le j \le m \}$ and $B=V(G) \setminus A$ such that no edge of $G$ has both its endpoints in the same set.

	We claim that there exists an assignment that satisfies $f$ if and only if the optimal solution $\mu$ for the instance of $\ind$-$\mmax$ (resp., $\ind$-$\msum$) has measure at most $1$ (resp., $\tau + m$).
		
	Suppose the existence of an assignment $\theta^* : X \to \{ \True, \False \}$ that satisfies $f$. Then we move a pebble starting on vertex $v_i$ to the vertex labelled $x_i$ if $\theta^*(x_i)=\True$, or to the vertex labelled $\bar{x}_i$ if $\theta^*(x_i)=\False$. Moreover, for each clause $(\ell^1_j \vee \ell^2_j \vee \ell^3_j)$ of $f$ at least one literal $\ell^i_j$ must be true w.r.t. $\theta^*$. This implies that the vertex labelled $\ell^i_j$ is adjacent only to $z_j$ and to a vertex $x_s$ where no pebble has been placed. We then move the pebble initially placed on vertex $z_j$ to the vertex $\ell^i_j$.

The resulting configuration of pebbles is an independent set for $G$ and each pebble has been moved to a node adjacent to its starting position. This implies that the maximum movement is $1$ and that the overall distance travelled by pebbles is $\tau + m$.
		
	Conversely, suppose the existence of an optimal solution $\mu^*$ for the instance of $\ind$-$\mmax$ (resp., $\ind$-$\msum$) that has measure equal to $1$ (resp., $\tau+m$).
	Notice that for every pebble $p$ initially placed on a vertex labelled $v_i$, $\mu^*(p) \in \{ x_i, \bar{x}_i \}$ must hold, otherwise either $\mu^*$ would be unfeasible or $c(\mu^*)$ would be greater than $1$ (resp., $\tau+m$, since $\tau+m$ pebbles need to be moved). Similarly for every pebble $p$ initially placed on a vertex labelled $z_j$, $\mu^*(p) \in \{ \ell^j_1, \ell^j_2, \ell^j_3 \}$ must hold.
		
	We construct an assignment $\theta$ for the variables in the following manner: if there is a pebble on the vertex labelled $x_i$ we set $\theta(x_i)= \mathrm{\True}$, otherwise there must be a pebble on $\bar{x}_i$ and we set $\theta(x_i) = \False$.

	We now show that $\theta$ is, indeed, an assignment that satisfies $f$. 		
	For each clause	$(\ell^1_j \vee \ell^2_j \vee \ell^3_j)$ the pebble placed on $z_j$ has been moved to a vertex $\ell^i_j$ corresponding to one of the three literals.
	This implies that no pebble has been moved to the unique vertex in $\cup_{i=1}^t\{x_i, \bar{x}_i\}$ that is adjacent to vertex $\ell^i_j$.
	The above implies that the variable corresponding to literal $\ell^j_i$ has been set to the value that satisfies $\ell^j_i$ and thus the whole clause is satisfied.
	\qed\end{proof}
	
Apparently, the above technique cannot be straightforwardly adapted to $\ind$-$\mnum$, and so we leave this as an interesting open problem.

	\subsection{Approximability of $\ind$-$\mmax$}
	Actually, as shown in Theorem \ref{thm:ind_max_sum_hard_bipartite}, $\ind$-$\mmax$ is hard already when the cost of an optimal solution is 1. This immediately implies the following:
	\begin{corollary}
		\label{cor:ind_max_2-inapx}
		$\ind$-$\mmax$ on bipartite graphs is not approximable in polynomial time within a factor of $2-\epsilon$ for any positive $\epsilon$, unless $\P=\NP$.
	\end{corollary}
	
	We now show that this bound is tight, by providing a polynomial-time solution, which is optimal unless an additive term of $1$, to $\ind$-$\mmax$ on any class of graphs where a maximum independent set can be found in polynomial time, e.g., perfect graphs (which include bipartite graphs), interval graphs, and claw-free graphs.

	Given a graph $H$ and a subset of vertices $A \subseteq V(H)$ we will denote the open neighbourhood of $A$ by $N_H(A)=\left\{ v \in V(H) : \exists u \in A \mbox{ s.t. } (u,v) \in E(H) \right\}$. Moreover we will denote the closed neighbourhood of $A$ by  $N_H[A]=A \cup N_H(A)$.
	
	Let $U^*$ be a maximum independent set of $G$, the following lemma holds:
	\begin{lemma}
		\label{lemma:independent_set_neighborhood}
		For each independent set $U$ of $G$ it is true that $|U^* \cap N_G[U]| \ge |U|$.
	\end{lemma}
	\begin{proof}
		By contradiction, let $|U^* \cap N_G[U]| < |U|$ then $U^\prime = \left( U^* \setminus N_G[U] \right) \cup U = \left( U^* \setminus \left( U^* \cap N_G[U] \right) \right) \cup U$ is an independent set of $G$ and $|U^\prime| > |U^*|$.
	\qed\end{proof}

	To prove the next lemma we use the following well known result:
	\begin{theorem}[Hall's Matching Theorem \cite{hall1935representatives}]
		\label{thm:hall_s_matching}
		Let $H=(V_1 \cup V_2, E)$ be a bipartite graph. There exists a matching of size $|V_1|$ on $H$ iff $|A| \le |N_H(A)|, \linebreak \forall A \subseteq V_1$.
	\end{theorem}

	\begin{lemma}
		\label{lemma:injective_function}
		For each independent set $U$ of $G$, there exists an injective function $f : U \to U^*$ such that $d_G(u,f(u)) \le 1$.
	\end{lemma}
	\begin{proof}
		Construct the bipartite graph $H=(U \cup U^*, E)$ where all vertices of $U$ are considered to be distinct from the ones in $U^*$ and $E=\{(u,v) \in U \times U^* : v \in N_G[\{u\}] \}$.
		Notice that, by construction, if two vertices $u,v$ are adjacent in $H$ either they are the same vertex or they are adjacent in $G$, i.e., $d_G(u,v)\le 1$. 		
		
		Now, Lemma \ref{lemma:independent_set_neighborhood} implies that, for every $A \subseteq U$, we have $|N(A)| = |U^* \cap N_G[A]| \ge |A|$.
		Hence, from Hall's Matching Theorem, there is a matching of size $|U|$ on $H$ (and thus the function $f$ exists).
	\qed\end{proof}		
	
	We are now ready to prove:
	\begin{theorem}
		\label{thm:mov_ind_max_opt_plus_one}
		There exists a polynomial-time algorithm for $\ind$-$\mmax$ which, for every class of graphs where the maximum independent set can be found in polynomial time, computes a solution $\tilde{\mu}$ such that $c(\tilde{\mu}) \le c^* + 1$ where $c^*$ is the measure of an optimal solution.
	\end{theorem}
	\begin{proof}
		The algorithm computes a maximum independent set $U^*$ of $G$ then, if $k > |U^*|$ it reports infeasibility, otherwise it optimally moves the pebbles towards (a subset of) the vertices of $U^*$. To do that, it proceeds as follows: for every value of $z$ from $0$ to $n-1$, it computes a solution $S_z$ for the maximum matching problem on the auxiliary bipartite graph $H=(A \cup U^*,E)$, where each vertex in $A$ is associated with a pebble, and $(p \in A,v \in U^*) \in E$ if and only if $d_G(\sigma(p), v) \le z$.
		
		Let $\tilde{z}$ be the first value of $z$ such that $|S_z|=k$, i.e., all the pebbles have been matched. For every pebble $p \in P$ set $\tilde{\mu}(p)=v$, where $v$ is the only vertex such that $(p,v)\in S_{\tilde{z}}$, and return $\tilde{\mu}$.
		Clearly $c(\tilde{\mu})=\tilde{z}$. Let $\mu^*$ be an optimal solution to $\ind$-$\mmax$ and let $U=\mu^*[P]$. By Lemma \ref{lemma:injective_function} there exists an injective function $f$ that maps every vertex of the independent set $U$ to an adjacent vertex of $U^*$. Thus, for every $p \in P$ we have $d_G(\sigma(p), \mu^*(p)) + d_G(\mu^*(p), f(\mu^*(p))) \le c^* + 1$, and therefore there exists a way to place all the pebbles on vertices of $U^*$ while travelling a maximum distance of at most $c^*+1$. This implies $\tilde{z} \le c^*+1$.
	\qed\end{proof}

	\subsection{Independency motion problems on trees and paths}

	Concerning independency motion problems on trees, we are able to devise two dynamic-programming algorithms for $\ind$-$\msum$ and $\ind$-$\mnum$, respectively.
	We note that it remains open to establish whether $\ind$-$\mmax$ on trees can be solved in polynomial time. We will, however, devise an efficient algorithm for solving $\ind$-$\mmax$ on paths.
	The details of these algorithm are presented in the following subsections.

	\subsubsection{Solving $\ind$-$\msum$ on trees}

	Our dynamic-programming algorithm relies on the property of optimal solutions shown by Lemma \ref{lemma:sum_tree_no_crossing} that is valid also for $\ind$-$\msum$.
	
	The algorithm first roots the tree $G$ at an arbitrary vertex $r \in V(G)$ to obtain the rooted tree $G_r$, then considers all the subtrees of $G_r$ in a bottom-up fashion.
	
	Here, the auxiliary problem we consider when a subtree $T_u$ of $G$ rooted at vertex $u$ is examined is that of placing exactly $j \le k$ pebbles on the vertices of $T_u$ such that their final positions induce an independent set of $G$ (and each pebble is placed on different vertex). Remind that $\eta(u)$ denotes the number of pebbles placed on $V(T_u)$.
	As for the connectivity problems, if $j < \eta(u)$ we want to move the $\eta(u)-j$ exceeding pebbles to the parent of $u$ (and thus outside $T_u$). In a similar manner, if $j > \eta(u)$, then the $j-\eta(u)$ missing pebbles are to be moved into $T_u$ from the parent of $u$, where we assume they are initially placed.
	
	We will denote by $\Opt[u,j]$ the cost of the optimal movement for this auxiliary problem. Notice that, in $\Opt[u,j]$, we are accounting for the cost of traversing all the edges of $T_u$ plus the edge from $u$ to its parent. To solve the original problem we need to find a solution corresponding to $\Opt[r,k]$. Clearly as $k=\eta(r)$ we do not have exceeding or missing pebbles, in this case. We now show how to combine these auxiliary problems.
	
	Let $\Opt^+[u,j]$ (resp., $\Opt^-[u,j]$) with $u \in V(G)$ and $0 \le j \le n$ be the value of an optimal solution to the auxiliary problem where exactly one pebble must be placed on $u$ (resp., no pebble can be placed on $u$). Notice that $\Opt[u,j] = \min \{ \Opt^+[u,j], \Opt^-[u,j] \}$. We will say that infeasible solutions have cost $+\infty$.

	If $u$ is a leaf we clearly have:
	\[
		\Opt^+[u,j] =
			\begin{cases}
				|\eta(u) - 1| & \mbox{if } j = 1; \\
				+\infty & \mbox{if } j=0 \mbox{ or } j \ge 2. \\
			\end{cases}			
	\quad \enskip		
		\Opt^-[u,j] =
			\begin{cases}
				\eta(u) & \mbox{if } j = 0; \\
				+\infty & \mbox{if } j \ge 1.
			\end{cases}
	\]
	
	If $u$ is not a leaf then we can either place a pebble on $u$ or not. If we do place it, the corresponding optimal value is:
	\[
		\Opt^+[u,j] =
		\begin{cases}
			+\infty & \mbox{if } j=0; \\
			|\eta(u) - j| + {\displaystyle \min_{\substack{0 \le j_1, \dots, j_{\ell} < j \\ \sum_{i=1}^{\ell} j_i = j-1}}} \left\{ \sum_{i=1}^{\ell} \Opt^-[v_i, j_i] \right\} & \mbox{if } j \ge 1,
		\end{cases}
	\]
	where $v_1, \dots, v_\ell$ is the set of children of $u$ in $G_r$. If $j$ is at least $1$ we place a pebble on $u$ and then consider all the possible ways of placing a total of $j-1$ pebbles on the subtrees rooted at the children of $u$, but not on the children themselves.

	If we do not place a pebble on $u$, then the optimal value is:	
	\[
		\Opt^-[u,j] = |\eta(u) - j| + {\displaystyle \min_{\substack{0 \le j_1, \dots, j_{\ell} \le j \\ \sum_{i=1}^{\ell} j_i = j }}} \left\{ \sum_{i=1}^{\ell} \Opt[v_i, j_i] \right\}.
	\]

\noindent
	This corresponds to the optimal way of placing $j$ pebbles on the subtrees rooted at the children of $u$.
	
	As we have previously shown, the optimal assignment to $j_1, \dots, j_{\ell}$ can be found in $O(\ell \cdot k)$ time by using dynamic programming.
	The value of the optimal solution is stored in $\Opt[r, k]$ and the solution itself can be reconstructed by proceeding backwards as shown in Section \ref{sec:con_sum_trees} for $\con$-$\msum$.

	As we can compute the values $\Opt[u,j]$, for a fixed $j$ and all $v \in V$, in $O(n \cdot k)$ time (the $\ell$-values sum up to $n-1$), the total time required by the algorithm is $O(n \cdot k^2)$.

	The above discussion, immediately leads to the following:	
	\begin{theorem}
		\label{thm:ind_sum_trees}
		$\ind$-$\msum$ on trees can be solved in $\O(n \cdot k^2)$ time.
	\end{theorem} %
	
\subsubsection{Solving $\ind$-$\mnum$ on trees}

	The algorithm is similar to the one for $\ind$-$\msum$: we root the tree $G$ at an arbitrary vertex $r \in V(G)$ to obtain the rooted tree $G_r$, then we consider all the subtrees of $G_r$ in a bottom-up fashion.

	Let $\gamma(u)$ be $1$ if at least one pebble is initially placed on the vertex $u \in V(G)$, and $0$ otherwise.
	As before, when the subtree $T_u$, rooted at the vertex $u$, of $G_r$ is examined, we consider an auxiliary problem where we want to place exactly $j \le k$ pebbles on the vertices of $T_u$ such that their final positions induce an independent set of $G$ (and each pebble is placed on a different vertex).

	We will measure the cost of a solution for this auxiliary problem by examining the number of pebbles placed on each vertex of $T_u$. Removing pebbles from a vertex costs nothing, while placing a pebble on a vertex costs $1$ if it comes from a different vertex.	A way to visualize this auxiliary problem is to imagine the tree $T_u$ where no pebbles have been placed, and	a pool of $j$ pebbles to be distributed on its vertices.
	If a vertex $v$ is such that $\gamma(v)=1$ then it can hold a single pebble for free, otherwise placing a pebble on $v$ will increase the overall cost by $1$.	

	We will denote by $\Opt[u,j]$ the cost of the optimal movement for this auxiliary problem. To solve the original problem we need to find a solution corresponding to $\Opt[r,k]$. We now show how to combine these auxiliary problems.
	
	Let $\Opt^+[u,j]$ (resp., $\Opt^-[u,j]$) with $u \in V(G)$ and $0 \le j \le n$ be the value of an optimal solution to the auxiliary problem where exactly one pebble must be placed on $u$ (resp., no pebbles can be placed on $u$). Notice that $\Opt[u,j] = \min \{ \Opt^+[u,j], \Opt^-[u,j] \}$. We will say that infeasible solutions have cost $+\infty$.

	If $u$ is a leaf of $G_r$ then we have:
	\[
		\Opt^+[u,j] =
		\begin{cases}
			+\infty	& \mbox{if } j=0 \mbox{ or } j \ge 2; \\
			1-\gamma(u)	& \mbox{if } j=1.
		\end{cases} \quad \quad
		\Opt^-[u,j] =
		\begin{cases}
			0		& \mbox{if } j=0; \\
			+\infty	& \mbox{if } j \ge 1.
		\end{cases}
	\]

	If $u$ is not a leaf then, with a reasoning similar to that we did for $\ind$-$\msum$ we have:
	\[
		\Opt^+[u,j] =
		\begin{cases}
			+\infty	& \mbox{if } j=0; \\
			1-\gamma(u) + {\displaystyle \min_{\substack{0 \le j_1, \dots, j_{\ell} < j \\ \sum_{i=1}^{\ell} j_i = j-1}}} \left\{ \sum_{i=1}^{\ell} \Opt^-[v_i, j_i] \right\} & \mbox{if } j \ge 1,
		\end{cases}	
	\]
	and
	\[
		\Opt^-[u,j] =
		 \min_{ \substack{0 \le j_1, \dots, j_{\ell} \le j \\ \sum_{i=1}^{\ell} j_i = j} } \left\{ \sum_{i=1}^{\ell} \Opt[v_i, j_i] \right\},
	\]
	where $v_1, \dots, v_{\ell}$ are the children of $u$ in $G_r$.

	Again, the optimal assignment to $j_1, \dots, j_\ell$ can be found in $O(\ell \cdot k)$ time by using dynamic programming. The value of the optimal solution is stored in $\Opt[r, k]$ and the solution itself can be reconstructed by proceeding backwards.
	As we can compute the values $\Opt[u,j]$, for a fixed $j$ and all $v \in V$, in $O(n \cdot k)$ time (the $\ell$-values sum up to $n-1$), the total time required by the algorithm is $O(n \cdot k^2)$.

	The above discussion immediately leads to the following:
	\begin{theorem}
		\label{thm:ind_num_trees}
		$\ind$-$\mnum$ on trees can be solved in $\O(n \cdot k^2)$ time.
	\end{theorem}

\subsubsection{Solving $\ind$-$\mmax$ on paths}

	In this section we concentrate on $\ind$-$\mmax$, and we devise an efficient algorithm for solving $\ind$-$\mmax$ on the special case where $G$ is a path. We start by proving the following:
	
	\begin{lemma}
		\label{lemma:swap}
		Let $s_1, s_2, e_1, e_2 \in \mathbb{N}$ where $s_1 \le s_2$ and $e_2 \le e_1$, we have \linebreak $\max\left\{ | s_1 - e_1|, |s_2 - e_2| \right\} \ge \max\left\{ | s_1 - e_2|, |s_2 - e_1| \right\}$.
	\end{lemma}
	\begin{proof}
		Suppose that $| s_1 - e_2| \ge |s_2 - e_1|$.
		If $s_1 \ge e_2$ then $|s_2 - e_2| \ge |s_1 - e_2|$, else if $s_1 < e_2$ then $|s_1 - e_1| = e_1 - s_1 \ge  e_2 -s_1 = |s_1 - e_2|$.
		
		In a similar manner suppose $|s_1 - e_2| \le |s_2 - e_1|$.
		If $s_2 \ge e_1$ then $|s_2 - e_2| \ge |s_2 - e_1|$, else if $s_2 < e_1$ then $|s_1 - e_1| = e_1 - s_1 \ge e_1 - s_2 = |s_2 - e_1|$.
	\qed\end{proof}

	Let $r$ be any endpoint of $G$, with a little abuse of notation we will identify the vertices of $G$ by their distance from $r$: the vertex $i$ with $0 \le i < n$ will be the unique vertex $u$ such that $d_G(r, u)=i$.
	
	 Let $p_1, \dots, p_k$ be an ordering of the pebbles such that $\sigma(p_i) \le \sigma(p_{i+1})$ for $i=1,\dots,k-1$.
	The following lemma shows that there exists an optimal solution where the ordering of pebbles is preserved, i.e., no edge of $G$ is traversed by two pebbles going in opposite directions.\footnote{This is not true if $G$ is a tree: take the tree $T$ with vertex set $\{1, 2, \dots, 11\}$ and edge set $\{ (i,i+1) : 1 \le i \le 4 \} \cup \{ (2,6), (2,7), (3,8),(3,9), (4,10), (4,11) \}$, place $2$ pebbles on vertex $1$ and one pebble on each vertex in $\{6, 7, 8, 9, 10, 11\}$. It is easy to see that  any optimal solution has cost $3$ and that at least one edge must be traversed in opposite directions.}
	
	\begin{lemma}
		\label{lemma:no_crossing}
		If there exists a feasible solution $\mu$ for $\ind$-$\mmax$ on paths then there also exists a solution $\mu^\prime$ such that $\mu^\prime(p_i) < \mu^\prime(p_{i+1})$ for every $i=1,\dots,k-1$ and $c(\mu^\prime) \le c(\mu)$.
	\end{lemma}
	\begin{proof}
		Let $\mu$ be a feasible solution. We will show that whenever there are two pebbles $p_i$ and $p_{i+1}$ such that $\mu(p_i) > \mu(p_{i+1})$, they can be swapped without increasing the cost of $\mu$ to obtain another feasible solution.\footnote{The case $\mu(p_i) = \mu(p_{i+1})$ is clearly impossible.} If needed, the procedure can be repeated until all the pebbles are placed in the right order.
		
		 In order to prove the above, consider a new solution $\mu^\prime$ where $\mu^\prime(p_i)=\mu(p_{i+1})$, $\mu^\prime(p_{i+1})=\mu(p_i)$, and $\mu^\prime(p_j) = \mu(p_j)$ for every $j \not\in \{i,i+1\}$.
		Clearly $\mu^\prime$ is feasible (as we only swapped two pebbles) and requires at most the same maximum movement as $\mu$, because:
		\begin{multline*}
			\max\left\{ d_G(\sigma(p_i), \mu(p_i)), d_G(\sigma(p_{i+1}), \mu(p_{i+1})) \right\} \ge \\
			\max\left\{ d_G(\sigma(p_i), \mu(p_{i+1})), d_G(\sigma(p_{i+1}), \mu(p_{i})) \right\},
		\end{multline*}
		where every distance $d_G(u,v)$ between two vertices $u$ and $v$ can be rewritten as $|u - v|$ and Lemma \ref{lemma:swap} holds.
	\qed\end{proof}

	The algorithm takes as input an instance of $\ind$-$\mmax$ plus a non-negative integer $z$ and exploits the previous property to compute a solution of cost at most $z$, if it exists. The idea is simple: the pebbles are moved towards the endpoint $r$ of the path as most as possible while preserving the feasibility of the solution and the constraint on the cost. The pseudocode is given below.

	\begin{algorithm}[!ht]
	\DontPrintSemicolon
	\SetKwInOut{Input}{Input}
	\SetKwInOut{Output}{Output}
	\SetKw{Break}{break}	

	\caption{Algorithm for $\ind$-$\mmax$\label{alg:IndMax}}
	
	\Input{An instance $\langle G, P, \sigma \rangle$ of $\ind$-$\mmax$ on paths, a non-negative integer $z$;}
	\Output{A solution $\mu$ such that $c(\mu) \le z$, iff such a solution exists.}
	\BlankLine

	$p_1, \dots, p_k \gets$ pebbles sorted by distance from $r$ \;

	$next \gets 0$ \;
	\For{$j \gets 1$ \KwTo $k$}
	{
		$h \gets \max\{next, \sigma(p_j)-z\}$ \;
		\;
		\If{$h \ge n$ or $h > \sigma(p_j)+z$}
		{
			\Return No feasible solution exists
		}			
		\;
		$\mu(p_j) \gets h$  \;
		$next \gets h+2$ \;
	}
	\;
	\Return $\mu$
	\end{algorithm}

	\begin{lemma}
	\label{lemma:alg1_correctness}
	Algorithm \ref{alg:IndMax} correctly returns a solution with maximum movement at most $z$, if such a solution exists, or reports infeasibility if there is no such solution.
	\end{lemma}
	\begin{proof}
		First notice that all solutions returned by the algorithm are feasible as no two pebbles can be placed on adjacent vertices and, by construction, no pebble travels a distance greater than $z$. This implies that if there is not a feasible solution, the algorithm correctly reports infeasibility.
		
		Now we will prove that the algorithm correctly computes a solution when this exists.
		Let $\mu^*$ be a solution with maximum movement at most $z$ that preserves the ordering of the pebbles (by Lemma \ref{lemma:no_crossing} such a solution always exists).
		We will prove by induction that at the end of the $j$-th loop	 the first $j$ pebbles are correctly placed (they are on an independent set), and that for the $j$-th pebble $\mu(p_j) \le \mu^*(p_j)$ holds.

		Base case: the first pebble is placed on the first reachable node from $r$, thus $\mu(p_1) \le \mu^*(p_1)$ must hold.
		
		Inductive step: suppose that the property is true for the first $j \ge 1$ pebbles, we will prove that it is also true also for the $(j+1)$-th pebble.
		If the pebble $p_{j+1}$ is placed on vertex $\sigma(p_{j+1})-z$ then the property holds, since this is the vertex with the smallest distance from $r$ that is reachable by $p_{j+1}$.
		Otherwise $\mu(p_{j+1}) = \mu(p_j)+2$, and we have $\mu^*(p_{j+1}) \ge \mu^*(p_j) + 2 \ge \mu(p_j) + 2=\mu(p_{j+1})$.
	\qed\end{proof}

We are ready to give the following:

	\begin{theorem}
		\label{thm:ind_max_paths}
		$\ind$-$\mmax$ on paths can be solved in $\O(n + k \log n)$ time.
	\end{theorem}
	\begin{proof}
		From Lemma \ref{lemma:alg1_correctness} it follows that $\ind$-$\mmax$ on paths can be solved as follows: if $k > n$ then there is no solution, otherwise we can perform a binary search using Algorithm \ref{alg:IndMax} to find the optimal value of $k$ (or report infeasibility).
	
		As far as the time complexity is concerned, each invocation of the algorithm requires time $O(k)$ and at most $O(\log n)$ invocations are required.
		For Algorithm \ref{alg:IndMax} to work, vertices and pebbles must be sorted w.r.t. their distance from $r$, and this needs to be done only once and requires $O(n + k)$ time. The whole procedure then takes $O(n + k \log n)$ time.
	\qed\end{proof}
	
	\section{Clique motion problems}
	\label{sec:clique}

	In this section we prove that the problems $\clique$-$\mmax$, $\clique$-$\msum$, and $\clique$-$\mnum$ are \NP-hard. Then, we give a tight approximation algorithm for $\clique$-$\mmax$ that computes a solution that costs at most one more than the optimal solution. As of $\clique$-$\msum$ and $\clique$-$\mnum$, we show that the problems are not approximable within any factor smaller than $10\sqrt{5}-21$, unless $\P=\NP$, and we devise two $2$-approximation algorithms.
	
	Actually, we will show that any approximation for  $\clique$-$\msum$ or $\clique$-$\mnum$ implies an approximation with the same ratio for the \emph{minimum vertex cover} problem, which is known to be not approximable under $10\sqrt{5}-21$, unless $\P=\NP$ \cite{dinur2005hardness}. Moreover, if the unique game conjecture \cite{khot2002power} is true, then both approximation algorithms are tight as the corresponding problems are not approximable within any constant factor better than $2$ \cite{khot2008vertex}.

	Finally, we show that for classes of graphs where we can find a \emph{maximum-weight clique} in polynomial time, we can also solve $\clique$-$\mnum$ in polynomial time.
	
	\subsection{Approximability of $\clique$-$\mmax$}

	We prove the following:
	\begin{theorem}
		$\clique$-$\mmax$ is \NP-hard.
	\end{theorem}
	\begin{proof}
		We show a reduction from the problem of determining if there exists a \emph{dominating clique} in a graph $H$, which is known to be \NP-complete \cite{haynes1998domination}. A dominating clique is a subset of vertices $C \subseteq V(H)$ such that $C$ is both a clique and a dominating set for $H$.
		
		We construct an instance of $\clique$-$\mmax$ by setting $G=H$ and placing a pebble on each vertex of $G$. We claim that there exists a dominating clique in $H$ if and only if the optimal solution for $\clique$-$\mmax$ has measure $c^*$ at most $1$.
		
		Suppose that there exists a dominating clique $C$ in $H$. By definition, $C$ is also a dominating set of $G$. We define $\mu$ so that every pebble initially placed on a vertex $u \not\in C$ is moved to a vertex $v \in C$ such that $(u,v) \in E(G)$ (notice that such a vertex always exists). After their movement the pebbles are placed on a clique of $G$ and each pebble has travelled a distance of at most $1$.
		
		Now suppose that there exists a solution $\mu$ for $\clique$-$\mmax$ such that $c(\mu) \le 1$. Clearly $\mu[P]$ is a clique, and it is also a dominating set, since for each vertex $u \not\in \mu[P]$, there exists a vertex $v \in \mu[P]$ such that $(u,v) \in G$.
	\qed\end{proof}
	
	On the positive side, we have:
	\begin{theorem}
		It is possible to compute a solution $\tilde{\mu}$ for $\clique$-$\mmax$ such that $c(\tilde{\mu}) \le c^* + 1$ in polynomial time, where $c^*$ is the measure of an optimal solution.
	\end{theorem}
	\begin{proof}
		Consider the following algorithm: for each vertex $u \in V(G)$ construct a solution $\mu_u$ that moves all the pebbles to $u$ (i.e., set $\mu_u(p)=u, \; \forall p \in P$) and compute $c(\mu)$. Among the $n$ possible solutions choose the one of minimum cost and call it $\tilde{\mu}$. 

Now, let $\mu^*$ be an optimal solution and recall that, by definition, $c(\mu^*) = \max_{p \in P} \left\{ d_G(\sigma(p), \mu^*(p))  \right\}$.
		Let $C=\mu^*[P]$ be the clique where pebbles have been placed by $\mu^*$, and notice that when the above algorithm considers a node $u \in C$ we have:		
	$ d_G(\sigma(p), \mu_u(p)) \le d_G(\sigma(p), \mu^*(p)) + 1 \; \forall p \in P
	$, thus $c(\tilde{\mu}) \le c(\mu_u) \le c(\mu^*)+1$. The claim follows.
	\qed\end{proof}

	\subsection{Approximability of $\clique$-$\mnum$}

	We prove the following:
	\begin{theorem}
		\label{thm:2-apx-mov-clique-num}
		$\clique$-$\mnum$ is $2$-approximable and is not approximable within any constant factor smaller than $10\sqrt{5}-21$, unless $\P=\NP$.
	\end{theorem}
	\begin{proof}
		Let $\langle G, P, \sigma \rangle$ be an instance of $\clique$-$\mnum$ and let $\varphi(u) = |\{p \in P : \sigma(p)=u \}|$ be the number of pebbles that are initially placed on vertex $u \in V$.
		Let us assume that $\varphi(u) \le 1$ for every $u \in V$, i.e., no two pebbles are placed on the same vertex. We will show later that this assumption is not restrictive.
		Call $H$ the graph induced by the vertices $u$ with $\varphi(u)=1$.
		Let $\bar{H}$ be the complement graph of $H$ w.r.t. the edge set, that is the graph such that $V(\bar{H})=V(H)$ and $E(\bar{H}) = \{ (u,v) : u,v \in V(H) \wedge u \not= v \wedge (u,v) \not\in E(H) \}$.
		We will show that there exists a vertex cover $\mathcal{C}$ for $\bar{H}$ if and only if there exists a solution for the instance $\langle G, P, \sigma \rangle$ of $\clique$-$\mnum$ of cost $|\mathcal{C}|$.

		Let $\mathcal{C}$ be a vertex cover for $\bar{H}$, this implies that $Q=V(H)-\mathcal{C}$ is an independent set for $\bar{H}$, and therefore a clique for $H$ and $G$. We construct a solution $\mu$ for $\clique$-$\mnum$ by moving the $|\mathcal{C}|$ pebbles that are not yet placed on vertices in $Q$ to one of such vertices.
		Now let $\mu$ be a solution for the instance $\langle G, P, \sigma \rangle$ of $\clique$-$\mnum$ and let $Q=\{u \in V(H) : \exists p \in P \mbox{ s.t. } \sigma(p)=\mu(p)=u\}$. Notice that the cost of $\mu$ is exactly $k-|Q|=|V(H)-Q|$. Clearly $Q \subseteq \mu[P]$ is a clique for $G$ and $H$, therefore it is also an independent set for $\bar{H}$. This implies that $\mathcal{C}=V(H)-Q$ is a vertex cover for $\bar{H}$.

		From the above it follows that the cost of an optimal solution is equal to the size of the minimum vertex cover for $\bar{H}$.
		To approximate $\clique$-$\mnum$ we construct the graph $\bar{H}$, we compute a $2$-approximate minimum vertex cover $\mathcal{\tilde{C}}$, and then we reconstruct the solution $\mu$ as shown.

		If the previous assumption is not met, i.e., there exists at least a vertex on which two or more pebbles are placed by $\sigma$, a slight modification to the instance is needed before we can apply the previous approach. We modify the graph $G$ by replacing each vertex $u$ such that $\varphi(u)>1$ with a clique of size $\varphi(u)$. Each edge $e$ incident to $u$ is replaced by $\varphi(u)$ edges connecting every vertex of the clique to the other endpoint of $e$. Then, we modify the function $\sigma$ so that the $\varphi(u)$ pebbles that were placed on $u$ are assigned to each of the $\varphi(u)$ vertices of the corresponding clique. After the modifications, the cost of an optimal solution has not changed, moreover every solution for the modified instance can be easily reconverted to a solution for the original instance without increasing its cost. Hence, we have designed a $2$-approximation algorithm for $\clique$-$\mnum$.
			
		Concerning the inapproximability result, we prove it by contradiction: we show that if there exists an algorithm that approximates $\clique$-$\mnum$ with a ratio better than $10\sqrt{5}-21$, then this would allow to approximate minimum vertex cover with the same approximation ratio. Indeed, let $\bar{H}$ be the instance of minimum vertex cover. Now, add an isolated vertex $u$ and
call $G$ the complement of such a graph w.r.t. the edge set. Then, let $P=[|V(G)|-1]$, and let $\sigma$ be a function that places a single pebble on each vertex of $G$ except $u$.
		We can now compute an approximate solution for the instance $\langle G, P, \sigma \rangle$ of $\clique$-$\mnum$ and reconstruct a solution with the same cost for minimum vertex cover, as shown before.\footnote{Notice that the additional vertex $u$ is only instrumental to guarantee that $G$ is connected.}
	\qed\end{proof}

	We close this section by proving that $\clique$-$\mnum$ can be solved in polynomial time whenever we are able to solve the maximum-weight clique problem on a weighted variant of $G$.
	We remark that this is known to be the case of several classes of graphs, including perfect graphs. We refer the interested reader to \cite{bomze1999maximum} for a survey.

	\begin{theorem}
		An exact solution to $\clique$-$\mnum$ can be found in polynomial time on every class of graphs where a maximum-weight clique can be found in polynomial time.
	\end{theorem}
	\begin{proof}
		For every $v \in V$ let $\varphi(v)$ be the number of pebbles starting on vertex $v$, i.e., $\varphi(v)=|\{p \in P : \sigma(p)=v\}|$.
		
		We consider a graph $H=(V,E)$ that is a copy of $G$ where each vertex $v \in V$ has a weight $\varphi(v)$.	Then, we compute a maximum-weight clique $Q$ of $H$ and construct a solution $\mu$ for $\clique$-$\mnum$ by moving every pebble starting on a vertex in $V \setminus Q$ to an arbitrary vertex of $Q$.\footnote{W.l.o.g. we assume that $Q$ contains no vertices of weight $0$.} It is now easy to see that this is an optimal solution, since for every other solution $\mu^\prime$, we have:
		\[
			c(\mu^\prime) = \sum_{v \not\in \mu^\prime[P]} \varphi(v) =
			k-\sum_{v \in \mu^\prime[P]} \varphi(v) \ge k-\sum_{v \in Q} \varphi(v) =
			\sum_{v \not\in \mu[P]} \varphi(v) = c(\mu).
		\]		
	\qed\end{proof}		
		
\subsection{Approximability of $\clique$-$\msum$}

	We prove the following:
	\begin{theorem}
		$\clique$-$\msum$ is $2$-approximable and is not approximable within any constant factor better than $10\sqrt{5}-21$, unless $\P=\NP$.
	\end{theorem}
	\begin{proof}
		Let $\mu^*$ be an optimal solution to $\clique$-$\msum$. If $\mu^*$ moves all the pebbles (i.e., $\sigma(p) \not= \mu^*(p), \; \forall p \in P$) then we can compute a $2$-approximate solution $\tilde{\mu}$ by guessing a vertex $u \in \mu^*[P]$ and moving all the pebbles to $u$ (i.e., setting $\tilde{\mu}(p)=u, \; \forall p \in P$).  Indeed, we have:
		\[
			\frac{c(\tilde{\mu})}{c(\mu^*)} = \frac{\sum_{p \in P}d_G(\sigma(p), \tilde{\mu}(p))}{\sum_{p \in P}d_G(\sigma(p), \mu^*(p))} \le \frac{|P| + \sum_{p \in P}d_G(\sigma(p), \mu^*(p))}{\sum_{p \in P}d_G(\sigma(p), \mu^*(p))} \le 2
		\]
		where we used the fact that $d_G(\sigma(p), \tilde{\mu}(p)) \le d_G(\sigma(p), \mu^*(p))+1$, and that $d_G(\sigma(p), \mu^*(p)) \ge 1$, for every pebble $p \in P$.
		
		On the other hand, if there exists at least one pebble $p^\prime \in P$ such that $\sigma(p^\prime)=\mu^*(p^\prime)$, then we guess its starting vertex $u=\sigma(p^\prime)$.
		We call $P_0$ the set of pebbles whose starting vertex is $u$, $P_1$ the set of pebbles whose starting vertex is adjacent to $u$, and $P_{2}$ the set of pebbles that are initially placed on a vertex at distance $2$ or more from $u$.
		We then set $\tilde{\mu}(p)=u$ if $p \in P_0$ or $p \in P_{2}$. With a reasoning similar to the one of the previous case we can show that:
		\[
			\sum_{p \in P_{2}} d_G(\sigma(p), \tilde{\mu}(p))  \le |P_{2}| + \sum_{p \in P_{2}} d_G(\sigma(p), \mu^*(p)) \le 2 \sum_{p \in P_{2}} d_G(\sigma(p), \mu^*(p)) \mbox{.}
		\]
		 Concerning $P_1$, assume $P_1\not=\emptyset$, and so we need to compute $\tilde{\mu}$ for the pebbles in $P_1$. To do that, consider the instance $\langle H, P_1, \sigma \rangle$ of $\clique$-$\mnum$ where $H$ is the subgraph of $G$ induced by the vertices initially occupied by pebbles in $P_1 \cup \{ p^\prime \}$, and compute a $2$-approximate solution $\mu^\prime$ as shown in Theorem \ref{thm:2-apx-mov-clique-num}. Set $\tilde{\mu}(p)=\sigma(p)$ for every pebble $p \in P_1$ such that $\mu^\prime(p)=\sigma(p)$, and set $\tilde{\mu}(p)=u$ for the remaining pebbles in $P_1$.
		
		Clearly $\tilde{\mu}[P]$ is a clique for $G$ as the vertices in $\mu^\prime[P_1]$ are a clique for $G$, $u$ is adjacent to every vertex in $\mu^\prime[P_1]$, and $\tilde{\mu}[P] \subseteq \mu^\prime[P_1] \cup \{ u \}$.

		Notice that the cost of moving the pebbles in $P_1$ w.r.t. $\mu^*$ is greater than or equal to the cost of the optimal solution for the instance $\langle H, P_1, \sigma \rangle$ of $\clique$-$\mnum$. Moreover, the cost of moving the pebbles in $P_1$ w.r.t. $\tilde{\mu}$ is equal to the cost of $\mu^\prime$. From the above it follows that:
		\[
			\sum_{p \in P_1} d_G(\sigma(p), \tilde{\mu}(p)) \le 2 \sum_{p \in P_1} d_G(\sigma(p), \mu^*(p)) \mbox{.}
		\]
		
		\noindent Therefore, the overall cost of this approximated solution is:
		\[
			c(\tilde{\mu}) \le 2 \sum_{p \in P_1} d_G(\sigma(p), \mu^*(p)) + 2 \sum_{p \in P_{2}} d_G(\sigma(p), \mu^*(p)) \le 2 \,c(\mu^*) \mbox{.}
		\]

		Concerning the inapproximability result, take a graph $H=(V, E)$ and construct the graph $G$ by complementing $H$ w.r.t. the edge set and adding an additional vertex $v_0$ adjacent to every other vertex. Let $P=[|V(H)|]$ and let $\sigma$ be a function that places one pebble on each vertex of $G$ except $v_0$. We will show that, given any solution for $\clique$-$\msum$, it is possible to construct a vertex cover of $H$ having the same cost, and vice versa. This implies that any approximation algorithm for $\clique$-$\msum$ converts into an approximation algorithm for minimum vertex cover with the same approximation ratio, therefore no approximation algorithm with an approximation ratio less than $10\sqrt{5}-21$ can exist for $\clique$-$\msum$ \cite{dinur2005hardness}, unless $\P=\NP$.
	
	Let $\mathcal{C}$ be a vertex cover for $H$, then $V(H)-\mathcal{C}$ is an independent set for $H$ and a clique for $G$. The solution that moves all the pebble of $\mathcal{C}$ to $v_0$ and leaves the others on their starting position is feasible and has a cost of $\mathcal{C}$.
	
	Now let $\mu$ be a solution for the instance of $\clique$-$\msum$. We modify $\mu$ in such a way that every pebble that moves at least by $1$ is now moving to $v_0$. This modification cannot increase the cost of $\mu$. Then, let $Q= \mu[P] \setminus \{ v_0 \}$, and let $\mathcal{C}=V(H) \setminus Q$. The cost of $\mu$ is $|V(H)|-|Q|=|\mathcal{C}|$, and $Q \subseteq \mu[P]$ is a clique for $G$. Therefore $Q$ is also an independent set for $H$, and $\mathcal{C}$ is a vertex cover for $H$.
	\qed\end{proof}

	\section{$\stcut$ motion problems}
\label{sec:stcut}

	In this section we discuss (in)approximability results for $\stcut$-$\mmax$ and $\stcut$-$\msum$. Among the others, we provide an essentially tight approximation algorithm for $\stcut$-$\mmax$. Regarding $\stcut$-$\mnum$, establishing its tractability remains open, but we will show that approximating such a problem can be useful to approximate $\stcut$-$\msum$, as well. We start by proving the following:

\FloatBarrier

\begin{figure}[t]
	\centering
	\includegraphics[scale=1.5]{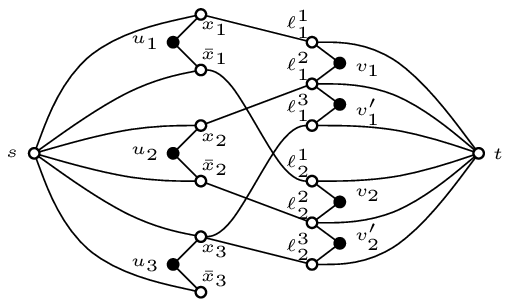}
	\caption{Instance of $\stcut$-$\mmax$ corresponding to the formula $(x_1 \vee x_2 \vee x_3) \wedge (\bar{x}_1 \vee \bar{x}_2 \vee x_3)$. Pebbles are placed on black vertices. All the edges not incident to a black vertex represent paths of length $h$ between their endpoints.}
	\label{fig:cut_max_hard}
\end{figure}

	\begin{theorem}
		\label{thm:cut_max_sum_hard_bipartite}
		$\stcut$-$\mmax$ and $\stcut$-$\msum$ are \NP-hard even if $G$ is a bipartite graph.
	\end{theorem}
	\begin{proof}
		We will show a reduction from the decision version of $\textsc{3-Sat}$.
	Let $X=\{x_1, \dots, x_\tau\}$ be the set of variables of the given formula $f$, and let $m$ be the number of clauses.
		Start with an empty graph $G$ and a set $P$ of $\tau + 2m$ pebbles, then construct an instance for $\stcut$-$\mmax$ and $\stcut$-$\msum$ as follows (see Figure \ref{fig:cut_max_hard} for an example):
		\begin{itemize}
			\item Add the two vertices $s$ and $t$.
			\item For each variable $x_i \in X$ add a path of three vertices to $G$. Label the middle node $u_i$ and the two endpoints $x_i$ and $\bar{x}_i$, respectively. Place a pebble on $u_i$.
			\item For each clause $(\ell^1_j \vee \ell^2_j \vee \ell^3_j)$ create a path of five vertices labelled, from one endpoint to another, $\ell^j_1, v_j, \ell^j_2, v_j^\prime, \ell^j_3$. Place a pebble on $v_j$ and one on $v_j^\prime$.
			\item For each literal $\ell^i_j$ of $f$, let $x_s$ be the corresponding variable. If $\ell^i_j$ is asserted add a new ``long'' path of length $h> k = t + 2 m$ between the vertices labelled $\ell^i_j$ and $x_s$ (so that $d_G(\ell^i_j,x_s)=h$). Otherwise add a new ``long'' path between the vertices labelled $\ell^i_j$ and $\bar{x}_s$.
			\item Connect $s$ to every vertex $x_i$ and to every vertex $\bar{x}_i$ with a new ``long'' path of length $h$.
			\item Connect $t$ to every vertex $\ell^j_i$ with a new ``long'' path of length $h$.
		\end{itemize}

		\noindent Notice that $G$ is bipartite as every cycle in $G$ must pass trough an even number of long paths (paths of length $h$) and an even number of other edges.
		
		We will show that there exists a truth assignment satisfying $f$ if and only if an optimal solution for $\stcut$-$\mmax$ (resp., $\stcut$-$\msum$) has cost at most $1$ (resp., $k$).
		
		Suppose that there exists a truth assignment $\theta^* : X \to \{ \True,\False \}$ that satisfies $f$. Then, for each $x_i \in X$ we move the pebble placed on $u_i$ to $x_i$ if $\theta^*(x_i)=\True$, and to $\bar{x}_i$ if $\theta^*(x_i)=\False$. Moreover, for each clause $(\ell^1_j \vee \ell^2_j \vee \ell^3_j)$ there are at most $2$ literals whose truth values are false w.r.t. $\theta^*$. We move one or two of the pebbles placed on $v_j$ and $v^\prime_j$ to the nodes of $G$ corresponding to the those literals.
				
		Notice that in this way each pebble travels at most a distance of $1$ (so the sum of the distances is at most $k$).  To prove that the new positions of the pebbles induce an \emph{s-t}-cut in $G$, consider any path $\pi(s,t)$ in $G$ between $s$ and $t$. We will show that $\pi(s,t)$ has been ``blocked'', i.e, there exists at least one pebble that has been moved on some vertex of $\pi(s,t)$.
		
		The path $\pi(s,t)$ must contain (as subpath) at least one path of length $h$ that connects a vertex $x \in \{x_i, \bar{x}_i\}$ for (some $i$) to a vertex $\ell \in \{\ell^1_j, \ell^2_j, \ell^3_j\}$ (for some $j$).
		If a pebble has been placed on $x$ then $\pi(s,t)$ is blocked. Otherwise the literal represented by $\ell$ must be false and, by construction, a pebble has been placed on $\ell$.

		Now suppose that an optimal solution for $\stcut$-$\mmax$ (resp., $\stcut$-$\msum$) has cost at most $1$ (resp., $k$). Notice that every pebble must have travelled a distance of at most $1$. This is trivial for $\stcut$-$\mmax$, while for $\stcut$-$\msum$ it suffices to note that placing a pebble on an internal vertex of the paths of length $h$ blocks at most the same set of paths that are blocked by placing it on one of the endpoints. Moreover no pebble can traverse a whole path as $h > k$.
		
		We define $\theta(x_i)=\True$ if a pebble has been placed on $x_i$, and $\theta(x_i)=\False$ otherwise. For each clause $(\ell^1_j \vee \ell^2_j \vee \ell^3_j)$ there exists at least one vertex $\ell \in \{\ell^1_j, \ell^2_j, \ell^3_j \}$ such that no pebble has been moved to $\ell$.
		Let $x_i$ be the variable associated with $\ell_i$, and let $x$ be the vertex labelled $x_i$ if $\ell = x_i$, or the vertex labelled $\bar{x}_i$ if $\ell = \bar{x}_i$. The pebble placed on $u_i$  must have been moved to $x$, otherwise there would exist a path from $s$ to $t$ passing through $x$ and $\ell$, and therefore the clause is satisfied.
	\qed\end{proof}		

	We now show that $\stcut$-$\mmax$ and $\stcut$-$\msum$ are actually very hard to approximate:
	\begin{theorem}
		$\stcut$-$\mmax$ and $\stcut$-$\msum$ are not approximable within a factor of $n^{1-\epsilon}$ for every $\epsilon>0$, unless $\P=\NP$. This also holds for bipartite graphs.
	\end{theorem}
	\begin{proof}
		As shown in the proof of Theorem \ref{thm:cut_max_sum_hard_bipartite}, it is possible to construct instances of $\stcut$-$\mmax$ (and $\stcut$-$\msum$) such that the optimal solution has measure at most $1$ (resp., $k=\tau + 2m$) if and only if a boolean formula in conjunctive normal form with three literals per clause, $\tau$ variables, and $m$ clauses is satisfiable.
		Moreover, any solution with cost $z$ such that $\tau+2m < z < h$ (recall that $h$ is the length of the ``long'' paths) can be easily transformed into a solution of cost at most $\tau + 2m$ by moving every pebble $p$ that has been placed on one long path to an appropriate endpoint (i.e., the one adjacent to $\sigma(p)$).
		This implies that if $f$ is satisfiable, the measure $z^*$ of an optimal solution is at most $k$ for both the problems, while if $f$ is not satisfiable, $z^*$ is at least $h$.				 
				
		Then, given a formula $f$, construct an instance for $\stcut$-$\mmax$ (and $\stcut$-$\msum$) as in the proof of Theorem \ref{thm:cut_max_sum_hard_bipartite}, and suppose that there exists a polynomial-time algorithm that approximates the optimal solution within a factor of $n^{1-\epsilon}$, for some positive $\epsilon \le 1$. Let $\tilde{c}$ be the measure of such an approximate solution. Observe that it is possible to upper bound the number of vertices $n$ of the instance with the quantity $2 + 3\tau + 5m + 2\tau h + 6mh \le 11h(\tau+m) \le 11hk$.

		If $f$ is satisfiable, we have $\tilde{c} \le k n^{1-\epsilon}$, while if $f$ is not satisfiable we have $\tilde{c} \ge h$. If $k n^{1-\epsilon} < h$ holds then it is possible to decide $\textsc{3-Sat}$ by running the approximation algorithm and looking the measure of the approximate solution.
		This can be guaranteed by choosing $h > (11k)^{2/\epsilon}$, as we have
		$h^\epsilon > 11 k^2 \ge 11^{1-\epsilon} k^{2-\epsilon} = k (11 k)^{1-\epsilon}$, which implies $h > k (11 h k)^{1-\epsilon} \ge k n^{1-\epsilon}$.
	\qed\end{proof}	

	By moving all the pebbles on a minimum $s$-$t$-cut of $G$ we can show that the inapproximability result provided above is essentially tight for $\stcut$-$\mmax$:
	\begin{theorem}
		$\stcut$-$\mmax$ is $d$-approximable in polynomial time, where $d < n$ is the diameter of $G$, while $\stcut$-$\msum$ is $(k \cdot d)$-approximable in polynomial time.
	\end{theorem}
	\begin{proof}
		The approximation algorithm is as follows: if the initial position of the pebbles already makes $s$ and $t$ disconnected, then we are done.
		Otherwise, compute (in linear time) a minimum $s$-$t$-cut of $G$, say $C$. If $|C| > k$ this implies that both problems have no solution, otherwise move the pebbles on the vertices of $C$.
		 The optimal measure must be at least $1$, while each of the $k$ pebbles travels a distance of at most $d$.
	\qed\end{proof}	

	We close this section proving a theorem which is useful in linking the approximability of $\stcut$-$\mnum$ to that of $\stcut$-$\msum$:
	\begin{theorem}
		If there exists a $\rho$-approximation algorithm for $\stcut$-$\mnum$ then there exists a $(\rho\cdot d)$-approximation algorithm for $\stcut$-$\msum$, where $d < n$ is the diameter of $G$.
	\end{theorem}
	\begin{proof}
		Consider an instance of $\stcut$-$\msum$ an let $c^*$ be the measure of an optimal solution. If $c^\prime$ is the measure of the optimal solution on $\stcut$-$\mnum$, then $c^* \ge c^\prime$. If $c^\prime = 0$ the claim is trivial, therefore we consider $c^\prime > 0$. Then,  a $\rho$-approximate solution for $\stcut$-$\mnum$ is a $\rho \cdot d$-approximate solution for $\stcut$-$\msum$, as each pebble travels a distance of at most $d$, and $\frac{\rho\cdot d \cdot c^\prime}{c^*} \le \frac{\rho\cdot d \cdot c^\prime}{c^\prime} = \rho \cdot d$.
	\qed\end{proof}

\section{Conclusions}
\label{sec:concl}
In this paper we have been concerned with the emerging field of (centralized) pebble motion problems. In particular, we have provided  approximability and inapproximability results --most of which were tight-- for several relevant variants studied in the literature. Among the issues we left open, the most prominent are those of establishing whether $\ind$-$\mmax$ on trees can be solved in polynomial time, and of settling the computational complexity of $\stcut$-$\mnum$.

Motion planning of devices in a constrained environment deserves a further deep investigation in several respects. Here we have limited our attention to vertex-to-vertex motion on unweighted, undirected graphs, and with the objective of achieving very basic configurations, but it is easy to imagine more challenging scenarios. For instance, pebbles could be deployed on a 2-dimensional environment, or even on a terrain. In this case, the planning task could become substantially more difficult, because of the setting in the continuum. On the other hand, a simplifying yet very interesting scenario is that in which the given graph is a 2-dimensional grid. At an intermediate stage, an intriguing variant is that in which the underlying graph is weighted. We also plan to look at other goal configurations, e.g., reaching a set of vertices inducing a 2-edge-connected subgraph, or connected subgraph with a bounded diameter, just to mention few.

\section*{Acknowledgements}

We would like to thank Feliciano Colella and Simone Galanti for insightful discussions on computational tractability of some of the problems studied in this paper.

\end{document}